\documentclass[centertags,12pt]{amsart}

\usepackage{amssymb}
\usepackage{amsmath}
\usepackage{bm}
\usepackage{amsthm}

\makeatletter
\renewcommand{\subsection}{\@startsection
{subsection}{2}{0mm}{\baselineskip}{-0.25cm}
{\normalfont\normalsize\em}}
\makeatother

\newtheorem{theorem}{Theorem}
\newtheorem{proposition}{Proposition}
\newtheorem{corollary}{Corollary}
\newtheorem{lemma}{Lemma}
{\theoremstyle{definition}

\newtheorem{example}{Example}}
\theoremstyle{remark}
\newtheorem{remark}{Remark}

\begin{document}

\title[Locally Recoverable codes from rational maps]{Locally Recoverable codes from rational maps}

\author{Carlos Munuera} 
\address{Department of Applied Mathematics, University of Valladolid, Avda Salamanca SN, 47014 Valladolid, Castilla, Spain}
\email{cmunuera@arq.uva.es}

\author{Wanderson Ten\'orio} 
 \address{Faculdade de Matem\' atica, Universidade Federal de Uberl\^ andia (UFU), Av. J. N. \' Avila 2121, 38408-902, Uberl\^ andia, MG , Brazil}
  \email{dersonwt@yahoo.com.br}

\begin{abstract}
We give a method to construct Locally Recoverable Error-Correcting codes.
This method is based on the use of rational maps between affine spaces. The recovery of erasures is carried out by Lagrangian interpolation in general and simply by one addition in some good cases. 
\end{abstract}


\maketitle


\section{Introduction}
\label{section1}

Error-Correcting codes are used to detect and correct any errors occurred during the transmission of information.  A {\em linear code}  of length $n$ over the finite field $\mathbb{F}_{q}$ is just a linear subspace $\mathcal{C}$ of $\mathbb{F}_{q}^n$. If $\mathcal{C}$ has dimension $k$ and minimum Hamming distance $d$ we say that it is a $[n,k,d]$ code.

{\em Locally Recoverable  (LRC) Error-Correcting} codes  were introduced in \cite{gopalan} motivated by the recent and significant use of coding techniques applied to distributed and cloud storage systems. Roughly speaking,  local recovery techniques enable us to repair  lost encoded  data by a "local procedure", which means by making use of small amount of data instead of all information contained in a codeword.

Let $\mathcal{C}$ be a $[n,k,d]$ code over $\mathbb{F}_{q}$. As a notation, given a  vector ${\bf{x}}\in\mathbb{F}_{q}^n$ and a subset $R\subseteq\{1,\dots,n\}$, we write ${\bf{x}}_R=\mbox{pr}_R({\bf{x}})$ and $\mathcal{C}_R=\mbox{pr}_R(\mathcal{C})$, where $\mbox{pr}_R$ is the projection on the coordinates of $R$. 
A coordinate $i\in\{ 1,\dots,n\}$ is {\em locally recoverable with locality $r$} if there is a {\em recovering set} $R(i)\subseteq \{1,\dots,n\}$ with $i\not\in R(i)$ and $\# R(i)= r$, such that for any codeword ${\bf{x}}\in\mathcal{C}$,  an erasure in position $i$  can be  recovered by using the information of ${\bf{x}}_{R(i)}$. That is to say, if 
for all ${\bf{x}},{\bf{y}}\in\mathcal{C}$, ${\bf{x}}_{R(i)}={\bf{y}}_{R(i)}$ implies $x_i=y_i$.
The code $\mathcal{C}$ has {\em all-symbol locality} $r$  if any coordinate is locally recoverable with locality at most $r$. We use the notation $([n, k], r)$ to refer to the parameters of such a code $\mathcal{C}$. 

The notion of local recoverability can be extended in two directions as follows:
Firstly it can be desirable to dispose of codes with multiple recovering sets for each coordinate \cite{availability}. 
The code $\mathcal{C}$  is said to have $t$ recovering sets if for each coordinate $i$  there exist disjoint sets $R(i,j)$, $j=1,\dots,t$, such that $\# R(i,j)\le r_j$ and the coordinate $x_i$ can be recovered from the coordinates of $\mbox{pr}_{R(i,j)}({\bf{x}})$ for all ${\bf{x}}\in \mathcal{C}$, as above. The existence of several recovering sets is known as the {\em availability} problem, as these codes make the recovering of loss data  better available.
Secondly we can recover more than one erasure at the same time. $\mathcal{C}$ is said to have the $(\rho,r)$ {\em locality property} if for each coordinate $i$ there is a subset $\bar{R}(i)$ containing $i$ such that  $\# \bar{R}(i)\le r+\rho-1$ and   $\rho-1$ erasures in ${\bf{x}}_{\bar{R}(i)}$ can be recovered from the remaining coordinates  of ${\bf{x}}_{\bar{R}(i)}$. 

It is clear that every code of minimum distance $d>1$ is in fact an LRC code  able to recover $\rho = d-1$ erasures, simply by taking ${\bar{R}(i)} = \{ 1,2, \dots, n \}$. But such recovering sets do not fit into the philosophy of local recovery. We are interested in codes $\mathcal{C}$ allowing smaller  recovering sets (in relation to the other parameters of $\mathcal{C}$). In this sense, we have the following Singleton-like bound: the locality $r$ of an $[n, k, d]$ code $\mathcal{C}$ obeys the relation (see \cite{gopalan})
\begin{equation} \label{Singleton-like}
\left\lceil\frac{k}{r}\right\rceil\le  n-k-d+2.
\end{equation}
Codes reaching equality are called {\em Singleton-optimal}  (or simply optimal) LRC codes, since they have the best possible relationship between these parameters. The number $n-k-d+2-\lceil{k}/{r}\rceil$ is the {\em optimal defect} of $\mathcal{C}$. Similarly we say that $\mathcal{C}$ is  {\em Singleton-almost-optimal} if, according to Eq. (\ref{Singleton-like}), no codes of type $([n,k+1,d],r)$ may exist, that is if ${(k+1)}/{r}> n-(k+1)-d+2$.

Slight modifications in the proof of Eq. (\ref{Singleton-like}) show that more generally, for $t=1,\dots,k$, we have
\begin{equation*} \label{Singleton-like-gen}
\left\lceil\frac{k-t+1}{r}\right\rceil\le  n-k-d_t+t+1
\end{equation*} 
where $d_1,\dots,d_k$ is the weight hierarchy of $\mathcal{C}$.
We deduce $r\le k$ and, from the duality of the weight hierarchy, $r\ge d^{\perp}-1$.

MDS codes (and Reed-Solomon codes in particular) satisfy $\mathcal{C}_R=\mathbb{F}_{q}^k$ for every set $R$ of $k$ positions, hence they have the largest possible locality $r=k$. Also, adding a parity check to $k$ information symbols we get a $([(k+1)t,kt,2],k)$ optimal code for all positive integer $t$ (which is MDS for $t=1$). Optimal codes that either are MDS or have the previous parameters  will be called {\em trivial}. 
The search for long nontrivial optimal codes is a challenging problem, and in fact  for most known optimal nontrivial codes, the cardinality of the ground field $\mathbb{F}_{q}$ is much larger than the code length $n$, \cite{largo1}. 

In \cite{barg1} a variation of Reed-Solomon codes for local recoverability purposes was introduced by Tamo and Barg. These so-called LRC RS codes are optimal and can have much smaller locality than RS codes themselves.  Its length is still smaller to the size of $\mathbb{F}_{q}$.
A classical way to obtain larger codes  is to use algebraic curves with many rational points. In this way LRC RS codes were extended by Barg, Tamo and Vladut \cite{barg2,barg3}, to the so-called  LRC Algebraic Geometry (LRC AG) codes, actually obtaining larger LRC codes.

In this article we propose a construction of LRC codes obtained from rational maps between affine spaces that, to a certain extent, generalizes the methods of  \cite{barg2,barg3,barg1}. The main idea is to consider a rational map $\phi:\mathbb{A}^m(\mathbb{F}_{q})\rightarrow \mathbb{A}^t(\mathbb{F}_{q})$ and a linear space $V$ of functions such that any $f\in V$ behaves as a univariate polynomial over each fibre of $\phi$. Under appropriate conditions it is possible to recover such polynomials  by using Lagrangian interpolation. Then the code $\mbox{ev}(V)$, where $\mbox{ev}$ is an evaluation map, is an LRC code whose recovering sets are the fibres of $\phi$.

The paper is organized as follows.  The general construction of our codes is developed  in Section 2. In the following two sections this construction is applied to some particular cases: to  algebraic curves in Section 3 and to other geometric sets in Section 4. Each case is illustrated with several examples.

\section{LRC codes from rational maps}
\label{section2}

The  construction of LRC codes  we propose relies on  rational maps between affine spaces,  extending ideas of  \cite{barg2,barg1}.  For the convenience of the reader, we begin  by briefly recalling those constructions.

\subsection{LRC codes from Algebraic Geometry}
\label{section2.1}

The construction of LRC RS codes is as follows \cite{barg1}: let $\mathbb{A}(\mathbb{F}_{q})$ be the affine line over $\mathbb{F}_{q}$ and let
$\mathcal{P}_1,\dots,\mathcal{P}_{t}\subset\mathbb{A}(\mathbb{F}_{q})$ be $t>1$ pairwise disjoint subsets of cardinality $r+1$ such that there exists a polynomial $\phi(x)\in\mathbb{F}_{q}[x]$ of degree $r+1$  which is constant over each $\mathcal{P}_i=\{P_{i,1},\dots,P_{i,r+1}\}$, $i=1,...,t$. Set $\mathcal{P}=\mathcal{P}_1\cup\dots\cup \mathcal{P}_{t}$ and $n=t(r+1)$. Fix an integer $k$ such that $r|k$ and $k+\frac{k}{r}\le n$, and consider the linear space of polynomials
$$
V=\left\{  \sum_{i=0}^{r-1} \; \sum_{j=0}^{k/r-1} a_{ij} \phi(x)^j x^i  \; : \; a_{ij}\in\mathbb{F}_{q}    \right\}.
$$
The code $\mathcal{C}$ is obtained by evaluation of $V$ at the points of $\mathcal{P}$, $\mbox{ev}_{\mathcal{P}}:V\rightarrow\mathbb{F}_{q}^n$, $\mbox{ev}_{\mathcal{P}}(f)=(f(P_{ij}), i=1,\dots,t, j=1,\dots,r+1)$. Since $\deg(f)\le k+\frac{k}{r}-2<n$ for all $f\in V$, the map $\mbox{ev}_{\mathcal{P}}$ is injective hence $\dim(\mathcal{C})=\dim(V)=k$. Furthermore given $f\in V$, for all $i=1,\dots,t$, there exists a polynomial $f_i(x)$ of degree $\le r-1$ such that $f_i(P_{ij})=f(P_{ij})$, $j=1,\dots,r+1$.  Such $f_i$'s can be computed through interpolation at any $r$ points of $\mathcal{P}_i$, so a recovering set for the coordinate corresponding to a point $P_{ij}$ is $R(i)=\mathcal{P}_i\setminus\{ P_{ij}\}$. Then $\mathcal{C}$ is a linear $([n,k],r)$ LRC optimal code.

The above method is extended to construct LRC AG codes as follows \cite{barg2,barg3}:  let $\mathcal{X},\mathcal{Y}$ be two algebraic  curves over $\mathbb{F}_{q}$ and let $\phi:\mathcal{X}\rightarrow \mathcal{Y}$ be a rational separable  morphism of degree $r+1$. Take a set $\mathcal{S}\subseteq \mathcal{Y} (\mathbb{F}_{q})$ of rational points with totally split fibres  and let  $\mathcal{P}=\phi^{-1}(\mathcal{S})$. Let $D$ be a rational divisor on $\mathcal{Y}$ with support disjoint  from $\mathcal{S}$ and  denote by
$\mathcal{L}(D)$ its associated Riemann-Roch space of dimension $m=\ell(D)$. 
By the separability of $\phi$ there exists  $y\in \mathbb{F}_{q}(\mathcal{X})$ satisfying $\mathbb{F}_{q}(\mathcal{X})=\mathbb{F}_{q}(\mathcal{Y})(y)$. Let 
$$
V=\left\{  \sum_{i=0}^{r-1} \; \sum_{j=1}^{m} a_{ij} f_j y^i  \; : \; a_{ij}\in\mathbb{F}_{q}    \right\}
$$
where $\{f_1,\dots,f_m\}$ is a basis of $\mathcal{L}(D)$. The LRC AG code $\mathcal{C}$ is defined as 
$\mathcal{C}=\mbox{ev}_{\mathcal{P}}(V)\subseteq \mathbb{F}_{q}^n$, with $n=\#\mathcal{P}$.
Note that $\mathcal{C}$ is a subcode of the algebraic geometry code $C(\mathcal{P},G)=\mbox{ev}_{\mathcal{P}}(\mathcal{L}(G))$, where $G$ is any divisor on $\mathcal{X}$ satisfying $V\subseteq \mathcal{L}(G)$. In particular $d(\mathcal{C})\ge d(C(\mathcal{P},G))\ge n-\deg(G)$.
Let $\mbox{ev}_{\mathcal{P}}(f)$, $f\in V$, be a codeword in $\mathcal{C}$.
Since the functions of $\mathcal{L} (D)$ are constant on each fibre  $\phi^{-1}(S)$, $S\in\mathcal{S}$,  the local recovery of an erased coordinate $f(P)$ of $\mbox{ev}_{\mathcal{P}}(f)$  can be performed by Lagrangian interpolation at the remaining $r$ coordinates of $\mbox{ev}_{\mathcal{P}}(f)$ corresponding to points in the fibre $\phi^{-1}(\phi(P))$ of $P$.  LRC AG codes with availability have been studied in \cite{matthews}. For a complete reference on algebraic geometry codes, see \cite{RM}.

\begin{theorem}\cite[Thm. 3.1]{barg3} 
If $\mbox{ev}_{\mathcal{P}}$ is injective on $V$ then $\mathcal{C}\subseteq \mathbb{F}_{q}^n$ is a linear $([n,k,d],r)$ LRC code with parameters 
$n=s(r+1)$, $k=r\ell(D)$ and $d\geq n-\deg(D)(r+1)-(r-1)\deg(y)$.
\end{theorem}

\begin{example}\label{exampleTamo}
(Example 1 of \cite{barg1}). The polynomial $\phi(x)=x^3\in\mathbb{F}_{13}[x]$  is constant over the sets $\mathcal{P}_1=\{ 1,3,9\}$, $\mathcal{P}_2=\{ 2,6,5\}$ and $\mathcal{P}_3=\{4,10,12\}$. By the first construction we get optimal LRC codes of length $9$, locality $2$ and dimensions $k=2,4,6$.
The same codes can be obtained by the second construction from the curve $\mathcal{X}:x^3=y$ and the morphism $\phi=y:\mathcal{X}\rightarrow\mathbb{P}^1$.
\end{example}

\begin{example}\label{ExampleBarg}
(a)  Codes from the Hermitian curve $\mathcal{H}:x^{q+1}=y+y^q$ over $\mathbb{F}_{q^2}$ were studied  in \cite{barg3}.
Consider the morphism $\phi=x:\mathcal{H}\rightarrow\mathbb{P}^1$ of degree $q$.  Take $\mathcal{S}=\mathbb{A}(\mathbb{F}_{q^2})\subseteq \mathbb{P}^1$ and thus $\phi^{-1}(\mathcal{S})=\mathcal{H}(\mathbb{F}_{q^2})\backslash \{Q\}$, where $Q$ is the point at infinity of $\mathcal{H}$.  If  $D=l\infty$ then 
$V=\bigoplus_{i=0}^{r-1} \langle 1,x,x^2,\ldots,x^l\rangle y^i\subseteq \mathbb{F}_{q^2}(\mathcal{H})$.
We get codes of locality $r=q-1$ and parameters $n=q^3$, $k=r(l+1)$, $d\geq n-\deg(G)$, where 
$G=(-v_Q(x^ly^{r-1}))Q$ and $v_Q$ is the valuation  at $Q$.
\newline
(b) Codes from the Norm-Trace curve 
$x^{1+q+\dots+q^{u-1}}=y+y^q+\dots+y^{q^{u-1}}$ 
over $\mathbb{F}_{q^u}$ were studied by Ballico and Marcolla in \cite{ballico}. Following similar ideas  they found a family of LRC codes of locality $r=q^{u-1}-1$, length
$n=q^{2u-1}$ and dimension $k=(t+1)(q^{u-1}-1)$. 
\end{example}

\subsection{LRC codes  from rational maps}
\label{section2.2}

Let $\phi_1,\dots,\phi_t,\phi_{t+1}\in \mathbb{F}_{q}(x_1,\dots,x_m)$ be rational functions and let $\mathcal{A}\subseteq\mathbb{A}^m(\mathbb{F}_{q})$
be a subset in which none of these functions have poles. Then the map $\phi=(\phi_1,\dots,\phi_t):\mathcal{A}\subseteq\mathbb{A}^m(\mathbb{F}_{q})\rightarrow \mathbb{A}^t(\mathbb{F}_{q})$ is well defined. Take two sets  $\mathcal{S}\subseteq \mathbb{A}^t(\mathbb{F}_{q})$  and  $\mathcal{P}\subseteq\phi^{-1}(\mathcal{S})$ and define $r=\min_{S\in\mathcal{S}^*}\# \{ \phi_{t+1}(P) : P\in\mathcal{P}, \phi(P)=S\}-1$, where $\mathcal{S}^*=\{ S\in \mathcal{S} : \phi^{-1}(S)\neq \emptyset\}$. Thus, the function $\phi_{t+1}$ takes at least $r+1$ different values in the points of the fibre  $\phi^{-1}(S)$ for all $S\in\mathcal{S}^*$.  
If $r>0$, for  $i=0,\dots,r-1$,  consider a linear $\mathbb{F}_{q}$-space $V_i\subset \mathbb{F}_{q}[\phi_1,\dots,\phi_t]$,  and let
\begin{equation}\label{defV}
V=\bigoplus_{i=0}^{r-1} V_i\, \phi_{t+1}^i \subset  \mathbb{F}_{q}[\phi_1,\dots,\phi_{t+1}]  \subset\mathbb{F}_{q}(x_1,\dots,x_m).
\end{equation}
We define the code $\mathcal{C}=\mathcal{C}(\mathcal{P},V)$ as the image of the evaluation at $\mathcal{P}$ map $\mbox{ev}_{\mathcal{P}}: V\rightarrow \mathbb{F}_{q}^n$, where $n=\#\mathcal{P}$. Then $\mathcal{C}$ is a linear code of length $n$. If $\mbox{ev}_{\mathcal{P}}$ is injective then its dimension is $\dim(V)$. Furthermore $\mathcal{C}$ is an LRC code of locality $r$, where recovering is obtained through Lagrangian interpolation in each fibre of $\phi$: let $f\in V$ and suppose we want to recover an erasure at position $f(P)$. 
Let $S=\phi(P)$ and $\{  P_{1},\dots,P_{r+1}=P\}\subseteq \phi^{-1}(S)\cap \mathcal{P}$ a set in which  $\phi_{t+1}$ takes $r+1$ different values. Since all functions in each $V_i$ are constant in the fibre of $S$,  the restriction of $f$ to this fibre can be written as a polynomial 
$f_S=\sum_{j=0}^{r-1} a_j T^j $, that is $f(P_{j})=f_S(\phi_{t+1}(P_{j}))$ for all $j=1,\dots,\#\phi^{-1}(S)$. Since  $\phi_{t+1}$ takes at least $r+1$ different values in  $\phi^{-1}(S)$, the polynomial $f_S$ may be computed by Lagrangian interpolation from $ \phi_{t+1}(P_{1}),\dots,\phi_{t+1}(P_{r})$.

\begin{example} \label{exrational}
Consider the rational functions $\phi_1=x/y, \phi_2=(y-1)/z$ and the map $\phi=\phi_1:\mathcal{A}\subset\mathbb{A}^3(\mathbb{F}_3)\rightarrow \mathbb{A}^1(\mathbb{F}_3)$, where $\mathcal{A}=\{ (x,y,z)\in \mathbb{A}^1(\mathbb{F}_3) : yz\neq 0 \}$.  Let $\mathcal{S}=\mathbb{A}^1(\mathbb{F}_3)$. 
Then $\phi^{-1}(0)=\{ (0, 1, 1), (0, 1, 2), (0, 2, 1), (0, 2, 2) \}$, 
$\phi^{-1}(1)=\{ (1, 1, 1), (1, 1, 2), (2, 2, 1)$, $(2, 2, 2) \}$, 
$\phi^{-1}(2)=\{ (1, 2, 1), (1, 2, 2)$, $(2, 1, 1), (2, 1, 2)  \} $. Note that   
$ \phi_2(0, 1, 1)=\phi_2 (0, 1, 2)$; 
$\phi_2(1, 1, 1)=\phi_2 (1, 1, 2)$ and
$\phi_2(2, 1, 1)=\phi_2 (2, 1, 2)$.  We set 
$\mathcal{P}=\{ (0, 1, 2), (0, 2, 1)$, $(0, 2, 2); (1, 1, 2)$, $(2, 2, 1), (2, 2, 2); (1, 2, 1), (1, 2, 2), (2, 1, 2)\}$ and hence $r=2$. If we take the spaces of functions 
$$
V_1=\langle 1, \frac{x}{y} \rangle \oplus \langle 1\rangle \,\frac{y-1}{z}  \;  \mbox{ and } \;
V_2=\langle 1, \frac{x}{y},  \frac{x^2}{y^2}\rangle \oplus \langle 1, \frac{x}{y}\rangle \,\frac{y-1}{z} 
$$
we get a $([9,3,6],2)$ optimal LRC code $\mathcal{C}_1$ (from $V_1$) and a $([9,5,3],2)$ optimal code $\mathcal{C}_2$ (from $V_2$). 
Both  have length much larger than the cardinality of $\mathbb{F}_3$. In fact, the length and locality of these codes are the same as those of  Example \ref{exampleTamo}, although the cardinal of the ground field is considerably smaller. 
\end{example}

Defined in such generality, it is not possible to say too much about the parameters and behaviour of codes obtained by the aforementioned construction. In the subsequent sections we will discuss in more detail several particular classes of codes given by this method. As we shall see, in some of these cases the local recovery can be done by a simple checksum and does not require interpolation, which is simpler and faster.

\begin{example}
For both codes $\mathcal{C}_1$ and $\mathcal{C}_2$ of Example \ref{exrational}, a direct inspection shows that for any codeword $\mbox{ev}_{\mathcal{P}} (f)$, $f\in V_2$  (as defined in Example \ref{exrational}), the sum of their coordinates corresponding to each fibre of $\phi$ is zero. Thus a single erasure is corrected by one addition.
\end{example}

This construction also allows us to recover more than one erasure. Following the standard notation as in the Introduction, suppose we want to correct $\rho-1$ erasures in the coordinates corresponding to a fibre of $\phi$. For this, let $r=\min_{S\in\mathcal{S}^*}\# \{ \phi_{t+1}(P) : P\in\mathcal{P}, \phi(P)=S\}-\rho$, consider the space of functions $V$ as stated in equation (\ref{defV}), with respect to the new value of $r$, and define $\mathcal{C}=\mbox{ev}_{\mathcal{P}}(V)$.
The restriction of $f\in V$ to the points in a fibre of  $\phi$ is a polynomial of degree $r-1$. It can be computed from the information of any $r$ available coordinates among the  at least $r+\rho-1$ coordinates corresponding to points in this fibre for which $\phi_{t+1}$ takes different values.

\section{LRC codes from Algebraic Curves}
\label{section3}

\subsection{The construction for algebraic curves} 
\label{constructionAG}

Let $\mathcal{X}\subset \mathbb{P}^m(\mathbb{F}_{q})$ be a (algebraic, projective, absolutely irreducible) curve over $\mathbb{F}_{q}$. 
Take a rational point $Q\in \mathcal{X}(\mathbb{F}_{q})$. After a change of coordinates, if necessary, we can assume that $Q$ is a point at infinity.
The set $\mathcal{L}(\infty Q)=\cup_{s\ge 0}\mathcal{L}(sQ)$ is a finitely generated $\mathbb{F}_{q}$-algebra. Take rational functions $\phi_1,\dots,\phi_{t+1} \in \mathcal{L}(\infty Q)$ and
fix representatives of these functions (still denoted $\phi_1,\dots,\phi_t,\phi_{t+1}$ by abuse of notation).  Let $\mathcal{A}\subseteq\mathcal{X}(\mathbb{F}_{q})$ be the set of rational affine points on $\mathcal{X}$ and consider the map $\phi=(\phi_1,\dots,\phi_t):\mathcal{A}\subseteq\mathbb{A}^m(\mathbb{F}_{q})\rightarrow \mathbb{A}^t(\mathbb{F}_{q})$.
Let $\mathcal{S}\subseteq \mathbb{A}^t(\mathbb{F}_{q})$, $\mathcal{P}\subseteq\phi^{-1}(\mathcal{S})$ and $r$ as in the general construction.
For  $i=0,\dots,r-1$,  fix  linear $\mathbb{F}_{q}$-spaces $V_i\subset \mathbb{F}_{q}[\phi_1,\dots,\phi_t]$  and 
$$
V=\bigoplus_{i=0}^{r-1} V_i\, \phi_{t+1}^i \subset  \mathbb{F}_{q}[\phi_1,\dots,\phi_{t+1}] .
$$
Since $\phi_{t+1}\in\mathcal{L}(\infty Q)$, no point of $\mathcal{P}$  is a pole of $\phi_{t+1}$.
We define the code $\mathcal{C}=\mathcal{C}(\mathcal{P},V)$ as the image of the evaluation at $\mathcal{P}$ map $\mbox{ev}_{\mathcal{P}}: V\rightarrow \mathbb{F}_{q}^n$, where $n=\#\mathcal{P}$ is the length of $\mathcal{C}$. Note that the number of zeros of a function $f$ is related to the value $v_Q (f)$,  where $v_Q$ is the valuation at $Q$: the larger
this number, then the smaller the weight of $\mbox{ev}_{\mathcal{P}}(f)$ may be. This fact leads us to use the language of polytopes to define the space of functions $V$. For our purposes a polytope will be a set $\mathfrak{P}\subset\mathbb{N}_0^{t+1}$, where $\mathbb{N}_0$ stands for the set of nonnegative integers. Often we shall take $\mathfrak{P}$ as the convex hull of a set of points including  $(0,\dots,0)$ and $(0,\dots,0,r-1)$. 

Given a polytope $\mathfrak{P}$ we can consider the associated linear  $\mathbb{F}_{q}$-space $V(\mathfrak{P})\subseteq\mathbb{F}_{q}[\phi_1,\dots,\phi_t]$ spanned by the functions  $\{ \phi_1^{\alpha_1}\cdots\phi_{t+1}^{\alpha_{t+1}} \; :$ $ (\alpha_1,\dots,\alpha_{t+1})\in \mathfrak{P}\}$. In our case, 
for $i=1,\dots,t+1$, let $v_i=-v_Q(\phi_i)$. For $l\ge 1$ we consider the polytope in $\mathbb{N}_0^{t+1}$
$$
\mathfrak{P}=\mathfrak{P}(l)=\{  (\alpha_1,\dots,\alpha_{t+1})\in \mathbb{N}_0^{t+1} \; :  \; v_1\alpha_1+\cdots+v_{t+1}\alpha_{t+1}\le l, \; \alpha_{t+1}\le r-1  \}.
$$
If $V=V(\mathfrak{P}(l))$, then $V\subseteq\mathcal{L}(lQ)$ and hence $\mathcal{C}(\mathcal{P}, V)$ is a subcode of the usual AG code $C(\mathcal{P}, lQ)$. Then $\mathcal{C}(\mathcal{P}, V)$ is an LRC code of locality $r$. When $l<n=\#\mathcal{P}$ the evaluation map $\mbox{ev}_{\mathcal{P}}:\mathcal{L}(lQ)\rightarrow\mathbb{F}_{q}^n$ is injective, then so is
$\mbox{ev}_{\mathcal{P}}:V\rightarrow\mathbb{F}_{q}^n$,  and the dimension of $\mathcal{C}(\mathcal{P}, V)$ is $\# \mathfrak{P}$. The minimum distance of $\mathcal{C}(\mathcal{P}, V)$ is at least the minimum distance of  $C(\mathcal{P}, lQ)$, that is $n-l$.

\begin{remark}
Our construction of LRC codes from curves extends that given in  \cite{barg2,barg3}. The main difference is in the set $V$ of functions to be evaluated.  In the language of \cite{barg2,barg3}, rather than considering a unique divisor $D$, we may consider $r$ rational divisors $D_0,\dots,D_{r-1}$ on $\mathcal{Y}$,  the space of functions 
$V=\bigoplus_{i=0}^{r-1}\mathcal{L}(D_i)y^i$
and the code  $\mbox{ev}_{\mathcal{P}}(V)\subseteq \mathbb{F}_{q}^n$.  This change does  not affect the locality nor the recovering method. Using a sequence of divisors $D_0,\dots,D_{r-1}$, instead of a single divisor $D$ provides greater flexibility to the construction. Typically one can take $D_0\ge \dots\ge D_{r-1}$. Since codewords of smaller weight come from evaluation of functions $f\in\mathcal{L}(D_i)y^i$ with larger $i$, this strategy allows to increase the dimension of $\mathcal{C}$ without affecting the estimate on the minimum distance.   In Subsection \ref{Hermitian} we shall show several examples of this fact.
Recall that  codes of this type were already considered by Maharaj in \cite{maharaj} using the language of function fields. In that paper, it is shown that the smallest divisor $G$ satisfying $V\subseteq\mathcal{L}(G)$  is described as 
$G=\max \{(\phi^*(D_i)-\mbox{div}(y^i) : 0\leq i \leq r-1\}$,
where  $\phi^*:\mbox{Div}(\mathcal{Y})\to \mbox{Div}(\mathcal{X})$ is the pull-back map induced by $\phi$.
\end{remark}

In the rest of this section we present examples of LRC codes obtained from some types of curves which are well known in coding theory.

\subsection{LRC codes from the Klein quartic} 
\label{klein}

The {\em Klein quartic} is the curve $\mathcal{K}$ defined over $\mathbb{F}_8$  by the affine equation
$x^3y+y^3+x=0$. Codes from this curve have been extensively studied,  \cite{Klein}.
It has 24 rational points, being two of them at infinity, $R=(0:1:0)$ and $Q=(1:0:0)$. The Weierstrass semigroups at $R$ and $Q$ are both equal to $\langle 3,5,7\rangle$. The rational functions $\phi_1=x/y$ and $\phi_2=x/y^2$ have poles only at $Q$, with multiplicity 3 and 5 respectively.

\begin{example}
Let $\mathcal{A}$ be the set of affine rational points of $\mathcal{K}$ and consider the rational function $\phi_1: \mathcal{A} \to \mathbb{A}^1(\mathbb{F}_8)$.
The only ramified point of $\phi_1$ is $(0,0)$, hence $\#\phi_1^{-1}(S)=3$ for all $S\in \mathbb{A}^1(\mathbb{F}_8)$, $S\neq 0$. Take $\mathcal{S}= \mathbb{A}^1(\mathbb{F}_8)\setminus \{0 \}$ and $\mathcal{P}=\phi_1^{-1}(\mathcal{S})=\mathcal{A}\setminus \{(0,0)\}$. Since $\phi_2$  takes 3 distinct values in all fibres $\phi_1^{-1}(S)$, $S\in\mathcal{S}$, we have $r=2$.  The polytope
$
\mathfrak{P}=\mathfrak{P}(6)=\{  (\alpha,\beta)\in \mathbb{N}_0^{2} \; :  \; 3\alpha+5\beta \le 6, \; \beta\le 1  \}
$
leads to the space of functions $V=\langle 1,\phi_1,\phi_1^2\rangle \oplus \langle 1 \rangle\phi_2$
which, through evaluation at $\mathcal{P}$, produces a $([21,4,15],2)$ LRC code of optimal defect 2. Let us note that according to the Griesmer bound, the maximum possible minimum distance $d$ for a linear $[21,4]_8$ code is $d=16$, and therefore optimal $([21,4]_8,2)$ LRC codes cannot exist.
 
In the same way, by considering polytopes  $\mathfrak{P}(l)$, $1\le l\le 20$, we obtain LRC  $[n,k]$ codes, $1\le k\le 13$, of locality 2. Those of dimension $1,11,13$ are optimal.  For the rest, those  of odd dimension have defect 1 and those of even dimension have defect 2. The true minimum distances of these codes have been computed with the computer system Magma \cite{magma}. 
In the non-optimal cases, further modifications of the above polytopes can lead to codes with better parameters. For example, if we remove $(6,0)$ from $\mathfrak{P}(20)$, we get a $[21,12,4] $  almost-optimal code. The same happens for the other values of $l$.
\end{example}

\subsection{LRC codes from elliptic curves} 
\label{elliptic}

Let $\mathcal{X}$ be an elliptic curve over $\mathbb{F}_{q}$. If we assume $\mbox{char}(\mathbb{F}_{q})=p\neq 2,3$, then $\mathcal{X}$ has a plane model given by an affine equation of the form $y^2=x^3+ax+b$ and exactly one point at infinity, $Q$. Although the above theory may be applied to any elliptic curve, in this section, by way of example, we shall restrict ourselves to those of equations $y^2=x^3+B$. 
If $q\equiv 3$ (mod 3) then all of them have $q+1$ rational points \cite{Rendiconti}, and no interest to construct LRC codes by our method.

Assume $q\equiv 1$ (mod 3) and let $\xi$ be a primitive element of $\mathbb{F}_{q}$. Since $\omega=\xi^{(q-1)/3}$ is a cubic root of unity, for each of these curves $\mathcal{X}$, if $(a,b)\in\mathcal{X}(\mathbb{F}_{q})$ then also  $(\omega a,\pm b)\in\mathcal{X}(\mathbb{F}_{q})$. Thus 
the set $\mathcal{A}$ of rational affine points $(a,b)$ of $\mathcal{X}$ with $a\neq 0$ can be grouped in triples, $\{ (a,b),(\omega a,b),(\omega^2 a,b)\}$,  and $a=0$  happens at most for two points. We can obtain LRC codes of locality $r=2$ from $\mathcal{X}$. To this end, let  $\mathcal{A}$ be the set of points belonging to these triples,  $\phi=y:\mathcal{A}\subset\mathbb{A}^2\rightarrow\mathbb{A}^1$, $\mathcal{S}=\mathbb{A}^1$ and  $\mathcal{P}= \mathcal{A} \subseteq \phi^{-1}(\mathbb{A}^1)$. Let $n=\# \mathcal{P}\ge \#\mathcal{X}(\mathbb{F}_{q})-3$. Since $-v_Q(x)=2, -v_Q(y)=3$, for $l\ge 1$ we consider the polytope
$\mathfrak{P}(l)=\{ (\alpha,\beta)\in\mathbb{N}_0^2 \; : \; 2\alpha+3\beta\le l, \alpha\le 1\}$
of cardinality $\# \mathfrak{P}(l)=l-\lfloor(l-1)/3\rfloor$ and the linear space of functions 
$V(l)=\langle \{x^{\alpha} y^{\beta} \; : \; (\alpha,\beta)\in \mathfrak{P}(l)\}\rangle \subset \mathcal{L}(lQ)$.
Let $\mathcal{C}=\mathcal{C}(\mathcal{P},V(l))$ obtained by evaluation at $\mathcal{P}$ of all functions in $V(l)$. Since $\mathcal{C}(\mathcal{P},V(l))$ $\subseteq C(\mathcal{P},lQ)$,  if $l<n$ then $\mbox{ev}_{\mathcal{P}}$ is injective, so the dimension of  $\mathcal{C}$ is $\# \mathfrak{P}(l)$ and its minimum distance is $n-l$. A straightforward computation shows that having these parameters,  $\mathcal{C}$ is optimal when $l\equiv 0 \, \mbox{(mod $3$)}$ and almost-optimal   
when $l\equiv 2 \,\mbox{(mod $3$)}$.
 
In order to determine the length of these codes, we shall study the number of rational points on our curves.  
There exist six isomorphism classes of such curves, namely $\mathcal{X}_i:y^2=x^3+\xi^i$, $i=0,\dots,5$.  In \cite{Rendiconti} a fast algorithm to determine the cardinality of all of them can be found. However, we do not use this algorithm here and we will limit ourselves to using some elementary facts established in that article. Write $S_i=\# \mathcal{X}_i-q-1$.  Since $\mathcal{X}_0 \cong \mathcal{X}_3$,  $\mathcal{X}_0 \cong \mathcal{X}_3$,  $\mathcal{X}_0 \cong \mathcal{X}_3$ over $\mathbb{F}_{q^2}$, 
$\mathcal{X}_0 \cong \mathcal{X}_2\cong \mathcal{X}_4$, $\mathcal{X}_1 \cong \mathcal{X}_3\cong \mathcal{X}_5$  over $\mathbb{F}_{q^3}$, and all of them are isomorphic over $\mathbb{F}_{q^6}$, we deduce that
$S_0=-S_3,  S_1=-S_4,  S_2=-S_5$ and $S_0^2+S_1^2+S_2^2=6q$. 
Then any of these six curves $\mathcal{X}_i$ has at least $q+1+\sqrt{2q}$ points. If $q$ is an even power, then  any of these six curves $\mathcal{X}_i$ has $q+1+2\sqrt{q}$ points. Even in the case that $q$ is an odd power, numerical experiments show that in most cases  any of these six curves has either $q+1+2\sqrt{q}$ or $q+2\sqrt{q}$ points. In conclusion, when $q\equiv 1 \, \mbox{(mod $3$)}$ we can always find optimal LRC codes of locality $r=2$ and length $n\ge q+\sqrt{2q}-2$.

\begin{example}
Consider the field $\mathbb{F}_{13}$ and the primitive element $\xi=2$. The curve $\mathcal{X}_2:y^2=x^3+4$ has 20 affine points (the maximum possible), namely $(7,\pm 3), (8,\pm 3), (11,\pm 3)$;  $(4,\pm 4), (10,\pm 4), (12,\pm 4)$;  $(2,\pm 5), (5,\pm 5), (6,\pm 5)$ and $(0,\pm 2)$. We obtain LRC Elliptic codes of length $n=18$, locality $r=2$ and dimensions $k=1,\dots,12$. Those of odd dimension are optimal whereas those of even dimension are  almost-optimal. Remark that the length of these codes is twice the length of those of Example \ref{exampleTamo}.  
\end{example}

\subsection{LRC codes from Hermitian curves}
\label{Hermitian}

Among all curves used to obtain AG codes, Hermitian curves $\mathcal{H}:x^{q+1}=y^{q}+y$ over $\mathbb{F}_{q^2}$, are by far the most studied and best known. In this subsection we include some examples of LRC Hermitian codes. Since such codes have been treated also in \cite{barg2, barg3}, these examples will allow us to show how our construction can provide codes with better parameters than those obtained from the original construction of \cite{barg2, barg3}. In subsequent subsections we will study further properties of Hermitian codes, including the recovery process.  

\begin{example}\label{improvementk}
Consider the Hermitian code $\mathcal{C}$ treated in \cite[Section IV]{barg3} and Example \ref{ExampleBarg}(a). It is constructed over $\mathbb{F}_{q^2}$ as the evaluation code  $\mathcal{C}=\mbox{ev}_{\mathcal{P}}(V)$ from the Hermitian curve $\mathcal{H}$, the set $\mathcal{P}=\mathcal{H}(\mathbb{F}_{q^2})\setminus \{Q\}$, where $Q$ is the only point at infinity, and the space of functions $V=\bigoplus_{i=0}^{r-1} \langle 1,x,x^2,\ldots,x^l\rangle y^i$.   Then $r=q-1$.
$\mathcal{C}$ can be obtained also  from the rectangular polytope
$\mathfrak{P}=\{ (\alpha,\beta)\in\mathbb{N}_0^2 \; : \; \alpha\le l, \beta\le r-1\}$.
In order to compare this code with some others obtained through our construction  of Section \ref{constructionAG}, we
note that $-v_Q(x)=q, -v_Q(y)=q+1$ and so $\max\{ -v_Q(f) \; : \; f\in V\}=lq+(r-1)(q+1)$. This fact suggests to consider the polytope
$\mathfrak{P}'=\{ (\alpha,\beta)\in\mathbb{N}_0^2 \; : \; q\alpha+(q+1)\beta\le lq+(r-1)(q+1), \beta\le r-1\}$ which
leads to the space of functions
$V'=\bigoplus_{i=0}^{r-1} \langle 1,x,x^2,\ldots,x^{l_i}\rangle y^i$
with $l_i=l+r-i-1$, $i=0,\dots,r-1$, and the code $\mathcal{C}'=\mbox{ev}_{\mathcal{P}}(V')$. Note that $V,V'\subseteq \mathcal{L}((lq+(r-1)(q+1))Q)$, hence whenever $lq+(r-1)(q+1)<n=q^3$, both maps $\mbox{ev}_{\mathcal{P}}:V\rightarrow \mathbb{F}_{q}^n$ and $\mbox{ev}_{\mathcal{P}}:V'\rightarrow \mathbb{F}_{q}^n$ are injective. So $\dim (\mathcal{C})=(l+1)r$ and 
$\dim (\mathcal{C}')=\# \mathfrak{P}=\dim (\mathcal{C}) + ((r-1)+\dots+0)$. Thus there is a gain of $(r-1)+\dots+1=(q-1)(q-2)/2$ units in the dimension of $\mathcal{C}'$ without affecting the estimate on the minimum distance. For example, if $q=5$ then for any value of $l\le 21$, $\mbox{ev}_{\mathcal{P}}$ is injective, the dimension increases by 6 units and the optimal defect decreases by 8 units. 
\end{example}

\begin{example}\label{improvementd}
Consider again the Hermitian code $\mathcal{C}$ of Example \ref{ExampleBarg}(a).
Let us take here the opposite way to the previous Example \ref{improvementk} and  construct LRC Hermitian codes with improved minimum distance and at least the same dimension. Let $V$ and $\mathcal{C}$ as in Example \ref{improvementk}. Let $l'=l-\lfloor r(r-1)/2q\rfloor$. Consider now the space of functions 
$V'=\bigoplus_{i=0}^{r-1} \langle 1,x,x^2,\ldots,x^{l'_i}\rangle y^i$ 
with $l'_i=l'+r-i-1$, $i=0,\dots,r-1$, and the code  $\mathcal{C}'=\mbox{ev}_{\mathcal{P}}(V')$. Note that $l'=l-\lfloor r(r-1)/2q\rfloor$ implies $q(l'+1)+r(r-1)/2\ge q(l+1)$ so, according to the considerations made in Example \ref{improvementk}, $\dim(\mathcal{C}')\ge \dim(\mathcal{C})$. Furthermore, since $V\subseteq \mathcal{L}((lq+(r-1)(q+1))Q)$, $V'\subseteq \mathcal{L}((l'q+(r-1)(q+1))Q)$, the estimate on the minimum distance of $\mathcal{C}'$ is $(l-l')q\approx r(r-1)/2$ units larger than the estimate on the minimum distance of $\mathcal{C}$.
For example, if $q=5$ then $\dim(\mathcal{C})=5l+5$, $\dim(\mathcal{C}')=5l+6$ and the estimate on the minimum distance of $\mathcal{C}'$ is $5$ units larger than the estimate on the minimum distance of $\mathcal{C}$.
Numerical experiments using Magma to compute the true minimum distances of these codes show that this is  the case for most values of $l$. 
\end{example}

\subsection{AG Hermitian codes as LRC codes}

Let $\mathcal{C}=\mbox{ev}_{\mathcal{P}}(V)$ be an LRC code obtained from a curve $\mathcal{X}$, a point $Q$ and the sets $\mathcal{P}$ and $V$.  Our estimate on the minimum distance of $\mathcal{C}$ comes from the Goppa bound on the minimum distance of the smallest AG code $C(\mathcal{P},G)$  containing $\mathcal{C}$, which is probably much bigger than $\mathcal{C}$. Furthermore, no efficient decoding algorithms for  $\mathcal{C}$ are currently available, so that we are compelled to decode it as a subcode of $C(\mathcal{P},G)$,  bypassing  its true error-correcting capability. Consequently it is natural to ask when $\mathcal{C}$ is an AG code itself, that is when $V=\mathcal{L}(G)$. We cannot expect this to happen in general, but it is possible when the linear space $V$ is suitably chosen. In this subsection we discuss the case of Hermitian codes and the rational map $\phi=x$.
Note that a function $f\in\mathcal{L}(\infty Q)\subset\mathbb{F}_{q^2}(\mathcal{H})$ can be uniquely written as a polynomial $f\in\mathbb{F}_{q^2}[x][y]$ with $\deg_y(f)<q$. Such $f$ may appear in the set $V$ of an LRC code arising from $\mathcal{H}$  if $\deg_y(f)<q-1$.

\begin{proposition}
Let $\mathcal{H}$ be the Hermitian curve over $\mathbb{F}_{q^2}$. The space of functions
$
V=\bigoplus_{i=0}^{t} \langle 1,x,x^2,\ldots,x^{l_i}\rangle y^i\subseteq \mathbb{F}_{q^2}(\mathcal{H})
$,
with $t\le q-2$, is a Riemann-Roch space $\mathcal{L}(G)$ if and only if there exists $j$, $0\le j\le t+1$, such that 
$l_{t-i}=i$ if $i<j$ and  $l_{t-i}=i+1$ if $i\ge j$. In this case $G=sQ$, where $Q$ is the point at infinity of $\mathcal{H}$ and $s=(t+1)q+t-j$  if $j\le t$ or $s=t(q+1)$ if $j=t+1$.
\end{proposition}
\begin{proof}
The smallest Riemann-Roch space containing $V$ is $\mathcal{L}( sQ)$ with
$s=\max\{ -v_Q(x^{l_i}y^i) : i=0,\dots,t\}=\max\{ l_iq+i(q+1) :  i=0,\dots,t\}$.
Let us first assume $V=\mathcal{L}(sQ)$. Then for $i=0,\dots,t$, we have $l_{t-i}\le i+1$, since otherwise 
$-v_q(x^{l_{t-i}}y^{t-i})>-v_Q(x^{i+2}y^{t-i})>-v_Q(y^{t+1})$ so $y^{t+1}\in V$. In particular $l_t=0$ or $l_t=1$. A similar argument proves that for all $i=1,\dots,t$, we have $l_{i-1}\ge l_{i}+1$. From these two conditions we deduce that  there exists $j$, $0\le j\le t+1$, such that 
$l_{t-i}=i$ if $i<j$ and  $l_{t-i}=i+1$ if $i\ge j$. Conversely, it is simple to check that under these conditions we have $s=-v_Q(x^{l_{t-j}}y^{t-j})=(j+1)q+(t-j)(q+1)=(t+1)q+t-j$  if $j\le t$ and $s=-v_Q(y^{t})=t(q+1)$ if $j=t+1$. In either case it holds that  $\dim(V)=\ell(sQ)$, hence we get equality $V=\mathcal{L}(sQ)$.
\end{proof}

\begin{example}
Consider the  LRC codes obtained from the Hermitian curve with $q=4$ and the map $\phi=x$. For short let us restrict to those codes correcting one erasure per fibre ($t=2$). There are exactly four among them which are AG codes, namely the evaluation of the following spaces 
\begin{gather*}
V=\langle 1,x,x^2 \rangle  \oplus
\langle 1,x \rangle y \oplus
\langle 1 \rangle y^2
= \mathcal{L}(-v_Q(y^2)Q)=\mathcal{L}(10 Q) \\ 
V=\langle 1,x,x^2,x^3 \rangle  \oplus
\langle 1,x \rangle y \oplus
\langle 1 \rangle y^2
= \mathcal{L}(-v_Q(x^3)Q)=\mathcal{L}(12 Q) \\
V=\langle 1,x,x^2,x^3 \rangle  \oplus
\langle 1,x,x^2 \rangle y \oplus
\langle 1 \rangle y^2 
= \mathcal{L}(-v_Q(x^2y)Q)=\mathcal{L}(13 Q) \\
V=\langle 1,x,x^2,x^3 \rangle  \oplus
\langle 1,x,x^2 \rangle y \oplus
\langle 1,x \rangle y^2 
= \mathcal{L}(-v_Q(xy^2)Q)=\mathcal{L}(14 Q) .
\end{gather*}
\end{example}

\subsection{A simplified recovering method for LRC Hermitian and related codes}

In order to motivate this subsection, we begin with a simple example.

\begin{example}\label{checksum}
Let us consider the LRC Hermitian code $\mathcal{C}=\mbox{ev}_{\mathcal{P}}(V)$ constructed over $\mathbb{F}_{9}$ from the Hermitian curve $\mathcal{H}$, the map $\phi=x$, the set $\mathcal{P}=\mathcal{H}(\mathbb{F}_{9})\setminus \{Q\}$ where $Q$ is the only point at infinity of $\mathcal{H}$, and the space of functions
$V=\bigoplus_{i=0}^{r-1} \langle 1,x,x^2,\ldots,x^{l_i}\rangle y^i$, where $r=q-1=2$.  Each fibre $\phi^{-1}(a)$ is of type $\{(a,b_1),(a,b_2),(a,b_3)\}$ with $a^4=b^3_i+b_i$ for $i=1,2,3$. Since a function $f\in V$ can be written as $f=g_0(x)+g_1(x)y$, we have $f(a,b_i)=g_0(a)+g_1(a)b_i$ and thus
$$
f(a,b_1)+f(a,b_2)+f(a,b_3)=3g_0(a)+g_1(a)(b_1+b_2+b_3)=g_1(a)(b_1+b_2+b_3)=0
$$
because $b_1+b_2+b_3=0$ as it is the sum of the roots of the polynomial $y^3+y-a^4$ (the coefficient of $y^2$). Then the coordinate $f(a,b_1)$ of the codeword $\mbox{ev}_{\mathcal{P}}(f)$ can be recovered as $f(a,b_1)=-f(a,b_2)-f(a,b_3)$.   
\end{example}

In this subsection we shall show a family of LRC codes from curves with this property of recovering by checksum.
We first recall some identities regarding roots of univariate polynomials. 
Consider the elementary symmetric polynomials over a domain $A$, 
$
\sigma_1(z_1,\dots,z_d) = z_1+\cdots+z_d, 
\sigma_2(z_1,\dots,z_d)=z_1z_2+\cdots+z_{d-1}z_d, 
\dots,
\sigma_d(z_1,...,z_d) =  z_1\cdots z_d
$.
Write $\sigma_i=\sigma_i(z_1,...,z_d)$.
For an arbitrary monic polynomial $t(x)=x^d+t_{d-1}x^{d-1}+\cdots+t_0\in A[x]$, the elementary symmetric polynomials on the roots  $z_1,...,z_d$ of $t(x)$ are related to its coefficients by the  \emph{Vieta's formulae} \cite{Vieta}:
$\sigma_i=(-1)^i t_{d-i}$, $i=1,...,d$.
Consider now for $i\geq 1$ the multivariate polynomials
$
\pi_i=\pi_i(z_1,...,z_d)=z_1^i+\cdots+z_d^i\in A[z_1,\ldots,z_d],
$
which are related to the elementary symmetric polynomials by the \emph{Newton-Girard relations} \cite{Vieta}: we have $\pi_1=\sigma_1$ and  for each integer $i>1$, 
$$
\pi_i  = (-1)^{i-1}i\sigma_i  -\sum_{j=1}^{i-1} (-1)^j\pi_{i-j}\sigma_j .
$$
Let $p$ be the characteristic of $\mathbb{F}_{q}$. Recall that
a polynomial $\ell(x)\in \mathbb{F}_{q} [x]$ is {\em linearized} if the exponents of all its nonzero monomials are powers of $p$. In that case, any polynomial of the form $\ell(x)-\alpha$, $\alpha\in\mathbb{F}_{q}$, will be called {\em affine $p$-polynomial}.

\begin{lemma} \label{lemma} If $z_1,...,z_d$ are the roots of an affine $p$-polynomial of degree $d$ over $\mathbb{F}_{q}$, then $\pi_i(z_1,...,z_d)=0$ for all $i=1,\ldots,d-2$.
\end{lemma}
\begin{proof}
By induction.
Clearly $\pi_1=\sigma_1=0$. Assume $\pi_1=\dots=\pi_{i-1}=0$ for $i\le d-3$. Then $\pi_i=\pm i\sigma_i$. If $\sigma_i=0$ we have $\pi_i=0$. 
If $\sigma_i\neq 0$ then $d-i=p^h-i$ is a power of $p$. Thus $i\equiv 0 \mbox{ (mod $p$)}$ and again $\pi_i=0$. Therefore $\pi_1=\ldots=\pi_{d-2}=0$.
\end{proof}

Next we extend the idea of Example \ref{checksum} to a wide class of curves providing families of LRC codes allowing local recovery by a checksum.  The codes we propose come from the family of Artin-Schreier curves and thus they include LRC Hermitian codes. More precisely let $\mathcal{X}$ be an algebraic smooth curve defined over $\mathbb{F}_{q}$ by an equation of separated variables 
$u(x)=v(y)$, where $u$ and $v$ are univariate polynomials over $\mathbb{F}_{q}$ of coprime degrees and  $v$ is a separable linearized polynomial of degree $m=p^h$, with $p=\mbox{char}(\mathbb{F}_{q})$. 

$\mathcal{X}$ has exactly one point at infinity $Q$ which is the common pole of $x$ and $y$. The Weierstrass semigroup of $\mathcal{X}$ at $Q$ is generated by $-v_Q(x)=\deg(v)$ and $-v_Q(y)=\deg(u)$. Consider the rational map
$\phi=x:\mathcal{X}(\mathbb{F}_{q})\setminus\{Q\}\to \mathbb{A}^1$ and let $\mathcal{S}=\{ a\in \mathbb{F}_{q} \ : \ \mbox{there exists $b\in \mathbb{F}_{q}$  with} \ u(a)=v(b)\}$. By our assumptions on $v$, for each $a\in \mathcal{S}$ the function $x-a$ has $m$ distinct zeroes and thus the fibre $\phi^{-1}(a)$ consists of $m$ points. Let $\mathcal{P}=\phi^{-1}(\mathcal{S})$ and $n=\# \mathcal{P}=m \,\#\mathcal{S}$. For a positive integer $l<n$ we consider the polytope
$
\mathfrak{P}=\{ (\alpha,\beta)\in\mathbb{N}_0^2 \; : \; \deg(v)\alpha+\deg(u)\beta\le l, \beta\le m-2\}
$,
which leads to the space of functions $V=\bigoplus_{i=0}^{m-2} \langle 1,x,\ldots,x^{l_i}\rangle y^i$. Clearly it holds that $V\subseteq \mathcal{L}(l Q)$.  
Consider also the set of points $\mathcal{P}=\phi^{-1}(\mathcal{S})\backslash\{Q\}$. Then we have a code $\mathcal{C}=\mbox{ev}_{\mathcal{P}}(V)$ of length $n$, dimension $\# \mathfrak{P}$ and minimum distance $d\ge n-l$. According to the results stated in Subsection \ref{constructionAG}, $\mathcal{C}$ is an LRC code  with locality $r=m-1$. 

\begin{proposition} 
The linear code $\mathcal{C}$ described above allows a local recovery  based on one addition. 
\end{proposition}
\begin{proof}
Let $\phi^{-1}(a)=\{P_1,...,P_m\}$ be the fibre of $a\in\mathcal{S}$ and write $f_{a j}=f(P_j)$ the evaluation at  $P_j\in\phi^{-1}(a)$ of
$
f=\sum_{i=0}^{m-2} g_i y^i\in V 
$. 
The functions $g_i$ are constant over each fibre so we can write  $g_i=g_i(P_j)\in\mathbb{F}_{q}$ for any $P_j\in \phi^{-1}(a)$.  Then $f_{a j}= g_0 b_j^0+\cdots+g_{m-2} b_j^{m-2}$
with $b_j=y(P_j)$ and so
$$
\sum_{j=1}^m f_{a j}=mg_0+g_1\pi_1+\cdots+g_{m-2}\pi_{m-2}=g_1\pi_1+\cdots+g_{m-2}\pi_{m-2}
$$
where $\pi_i=\pi_i(b_1,\cdots,b_m)$ is the $i$-th Newton-Girard polynomial on the roots $b_1,\dots,b_m$ of $\ell(y)=v(y)-u(a)$. Since $\ell(y)$ is an affine $p$-polynomial, it follows from Lemma $\ref{lemma}$ that $\pi_1=\dots=\pi_{m-2}=0$. Therefore $f_{a1}+\cdots+f_{am}=0$
and the recovery of an erased symbol $f_{a j}$ is obtained through addition  of the remaining symbols in the fibre $\phi^{-1}(a)$.
\end{proof}

\begin{example}
Hermitian and Norm-Trace curves, among many others, belong to the family of curves we have considered. Thus LRC codes arising from them admit recovering of single erasures by one addition.
\end{example}

\section{LRC codes from other sets of points}
\label{othersets}

The same ideas and methods of the previous section can be applied to different geometric objects,  other than algebraic curves. In this section we will discuss three cases related to Affine Variety codes, Toric codes and Reed-Muller codes. All of them are well-known  in Coding Theory.

\subsection{Polytopes, polynomials and codes}
\label{polynomials}

Let us consider $\mathbb{F}_{q}[x_1,\dots,x_m]$ the ring of polynomials in $m$ indeterminates over the field $\mathbb{F}_{q}$.
A monomial in $\mathbb{F}_{q} [x_1,\dots,x_m]$ is called {\em reduced} if its degree  in each variable is at most $q-1$. A polynomial is reduced if it is a linear combination of reduced monomials. 
For any polynomial $f\in\mathbb{F}_{q}[x_1,\dots,x_m]$ there is a reduced polynomial  $f^*$ (obtained by reducing mod $q$ all exponents of $f$) such that $f(P)=f^*(P)$ for all $P\in\mathbb{A}^m(\mathbb{F}_{q})$. 
Thus, in what follows we shall restrict to consider reduced polynomials.

Given a polytope $\mathfrak{P}\subseteq [0,q-1]^m$ and two sets $\mathcal{A}\subseteq \mathbb{A}^m(\mathbb{F}_{q})$, $\mathcal{S}\subseteq \mathbb{A}^{t}(\mathbb{F}_{q})$, we can construct codes in the usual way. Let  $\phi: \mathbb{A}^m\rightarrow\mathbb{A}^{t}$ be the projection on $t$ coordinates and $\mathcal{P}=\phi^{-1}(\mathcal{S})$. We have the code $\mathcal{C}=\mathcal{C}(\mathcal{P}, V(\mathfrak{P}))$.  When these sets are properly chosen then $\mathcal{C}$ is an LRC code. The most obvious choices for a polytope are simplices and  hypercubes. Following the usual notation, given integers $l,l_1,\dots,l_m$, we write
\begin{eqnarray*}
\Delta(l) &=&\{ (\alpha_1,\dots,\alpha_m) \in \mathbb{N}_0^m \; : \;  \alpha_1+\cdots+\alpha_m\le l\}  ;  \\
\mathcal{H}(l_1,\dots,l_m) &=&\{ (\alpha_1,\dots,\alpha_m) \in \mathbb{N}_0^m \; : \; \alpha_i<l_i, \;  i=1,\dots,m\}.
\end{eqnarray*}
Note that when $\mathfrak{P}=\Delta(l)$ and $\mathcal{P}=\mathbb{A}^m(\mathbb{F}_{q})$ we obtain the Reed-Muller code $\mbox{RM}(l,m)$,  \cite{RM}. Next we shall state some properties of codes $\mathcal{C}(\mathcal{P}, V(\mathfrak{P}))$.

\begin{proposition}
Let $\mathfrak{P}\subseteq \mathcal{H}(q,\dots,q)\subseteq \mathbb{N}_0^m$ be a polytope,  $\mathcal{P}\subseteq  \mathbb{A}^m(\mathbb{F}_{q})$ a subset with $n$ points and $\mathcal{C}=\mathcal{C}(\mathcal{P}, V(\mathfrak{P}))$. Let $l$ be the maximum degree of a polynomial in $V(\mathfrak{P})$. Write $(q-1)m-l=\theta (q-1)+\mu$ with $0\le \mu<q-1$ and set $\delta=(\mu+1)q^{\theta}$. Then \newline
(a) the minimum distance of $\mathcal{C}$ is at least $\delta-q^m+n$; \newline
(b) if $\delta>q^m-n$ then $\mbox{ev}_{\mathcal{P}}:V(\mathfrak{P})\rightarrow\mathbb{F}_{q}^m$ is injective so $\dim(\mathcal{C})=\# \mathfrak{P}$.
\end{proposition}
\begin{proof}
Since  $l$ is the maximum degree of a polynomial in $V(\mathfrak{P})$, then $\mathcal{C}\subseteq \mbox{RM}(l,m)$, whose minimum distance is $\delta$,   \cite{RM}.
Furthermore $\mathcal{P}$ is obtained from $\mathbb{A}^m(\mathbb{F}_{q})$ by deleting $q^m-n$ points.
Both statements (a) and (b) are direct consequences of this fact.
\end{proof}  
  
The following result was proved for the case of Toric codes in \cite{prod}.

\begin{proposition}\label{producto}
Let $m_1, m_2$ be positive integers and $m=m_1m_2$. Consider the sets 
$\mathcal{P}_1\subseteq \mathbb{A}^{m_1}(\mathbb{F}_{q})$, $\mathcal{P}_2\subseteq \mathbb{A}^{m_2}(\mathbb{F}_{q})$, $\mathcal{P}= \mathcal{P}_1\times\mathcal{P}_2\subseteq \mathbb{A}^{m}(\mathbb{F}_{q})$ and  the polytopes $\mathfrak{P}_1\subseteq [0,q-1]^{m_1}$, 
$\mathfrak{P}_2\subseteq [0,q-1]^{m_2}$, $\mathfrak{P}=\mathfrak{P}_1\times \mathfrak{P}_2\subseteq [0,q-1]^{m}$.
Let $V_1, V_2, V$ be the linear spaces $V_1=V(\mathfrak{P}_1)$,  $V_2=V(\mathfrak{P}_2)$,  $V=V(\mathfrak{P})$. 
If both maps $\mbox{ev}_{\mathcal{P}_1}:V_1\rightarrow\mathbb{F}_{q}^{m_1}$ and $\mbox{ev}_{\mathcal{P}_2}:V_2\rightarrow\mathbb{F}_{q}^{m_2}$ are injective, then \newline
(a) the  map $\mbox{ev}_{\mathcal{P}}:V\rightarrow\mathbb{F}_{q}^{m}$ in injective; and \newline
(b) the minimum distances of the codes 
$\mathcal{C}(\mathcal{P}_1,V_1)$,  $\mathcal{C}(\mathcal{P}_2,V_2)$, $\mathcal{C}(\mathcal{P},V)$, satisfy
$d(\mathcal{C}(\mathcal{P},V))=d(\mathcal{C}(\mathcal{P}_1,V_1))\cdot d(\mathcal{C}(\mathcal{P},V))$.
\end{proposition}
\begin{proof}
Write $n_1=\mathcal{P}_1, n_2=\mathcal{P}_2, n=n_1n_2=\mathcal{P}$ and 
$d_1=d(\mathcal{C}(\mathcal{P}_1,V_1)), d_2=d(\mathcal{C}(\mathcal{P},V)), d=d(\mathcal{C}(\mathcal{P},V))$. A polynomial $f\in V$, $f\neq 0$, can be written as
$$
f=\sum_{\bm{\beta}\in\mathfrak{P}_2} f_{\bm{\beta}}(\bm{x})\bm{y}^{\bm{\beta}}
$$
where $\bm{\beta}=(\beta_1,\dots,\beta_{m_2})$, $\bm{y}^{\bm{\beta}}=y_1^{\beta_1}\dots y_{m_2}^{\beta_{m_2}}$ and
$f_{\bm{\beta}}\in V_1$. Denote by $Z(f)$ the set of zeroes of $f$.
(a) Since $\mbox{ev}_{\mathcal{P}_1}$ is injective, there exists a point $P_1\in\mathcal{P}_1$ such that $f(P_1,\bm{y})\neq 0$. Now  since $\mbox{ev}_{\mathcal{P}_2}$ is injective, there exists a point $P_2\in\mathcal{P}_2$ such that $f(P_1, P_2)\neq 0$ and so $\mbox{ev}_{\mathcal{P}}(f)\neq \bm{0}$. 
(b) Let $P_1\in\mathcal{P}_1$. If $P_1\in Z(\{  f_{\bm{\beta}}(\bm{x}) \, : \; \bm{\beta}\in \mbox{supp} (f)\})$ then $f(P_1,\bm{y})=0$ so $\# Z(f(P_1,\bm{y}))=n_2$. Observe that in this case we have $P_1\in Z(\sum f_{\bm{\beta}}(\bm{x}))\cap \mathcal{P}_1$. Since $\sum f_{\bm{\beta}}(\bm{x})\in V_1$ and $\mbox{ev}_{\mathcal{P}_1}$ is injective, this possibility happens for $s\le n_1-d_1$ points $P_1$.
 If $P_1\not\in Z(\{  f_{\bm{\beta}}(\bm{x}) \, : \; \bm{\beta}\in \mbox{supp} (f)\})$ then $f(P_1,\bm{y})$ has at most $n_2-d_2$ zeros $P_2\in\mathcal{P}_2$ because $\mbox{ev}_{\mathcal{P}_2}$ is injective. Thus
$$
\# Z(f) \le sn_2+(n_1-s)(n_2-d_2)=n-d_2(n_1-s)\le n-d_1d_2
$$
so $\mbox{wt}(f)\ge d_1d_2$ and  $d\le d_1d_2$.
Conversely if $f_1\in V_1, f_2\in V_2$ then $f_1f_2\in V$ and $\mbox{wt}(\mbox{ev}_{\mathcal{P}}(f))=\mbox{wt}(\mbox{ev}_{\mathcal{P}_1}(f_1)) \mbox{wt}(\mbox{ev}_{\mathcal{P}_2}(f_2))$, so $d\le d_1d_2$.
\end{proof}

\begin{corollary}\label{corprod}
Let $l_1,\dots,l_m \le q$ be positive integers and let $\mathfrak{P}$ be the hypercube $\mathfrak{P}=\mathcal{H}(l_1,\dots,l_m)$. Let  $\mathcal{P}=\mathcal{P}_1\times \cdots\times\mathcal{P}_m\subseteq \mathbb{A}^m(\mathbb{F}_{q})$ with $n_i=\# \mathcal{P}_i \ge l_i$ for all $i=1,\dots,m$. Then  $\mbox{ev}_{\mathcal{P}}:V(\mathfrak{P})\rightarrow\mathbb{F}_{q}^{m}$ in injective, so $\dim (\mathcal{C}(\mathcal{P}, V(\mathfrak{P})))=\#\mathfrak{P}$ and $d(\mathcal{C}(\mathcal{P}, V(\mathfrak{P})))=\prod_i (n_i-l_i+1)$.
\end{corollary}
\begin{proof}
$\mathcal{H}(l_1,\dots,l_m)$ is the product of $m$ segments $[0,l_i-1]$,  each of them leading to a Reed-Solomon code. The result follows from Proposition \ref{producto}.
\end{proof}

\subsection{LRC codes from Affine Variety codes}
\label{affinevariety}

Given $m$ positive integers $n_1,\dots,n_m$, such that $n_i|q-1$ for all $i$, set $v_i=(q-1)/n_i$ and consider the ideal $I$ of $\mathbb{F}_{q}[x_1,\dots,x_m]$ generated by the monomials $x_1^{n_1}-1,\dots,x_m^{n_m}-1$. Let $\mathcal{A}$ be the set of zeroes of $I$, $\mathcal{A}=Z(I)$. If $\xi$ is a primitive element of $\mathbb{F}_{q}$ then $\xi^{v_i}$ is a root of $x_i^{n_i}-1$, hence $\mathcal{A}$ consists of $n=n_1\cdots n_m$ points.  Furthermore, if we consider the projection map $\phi=(x_1,\dots,x_{i-1},x_{i+1},\dots,x_m): \mathbb{A}^m\rightarrow\mathbb{A}^{m-1}$, for every point $S\in\mathcal{S}=\phi(\mathcal{A})$ we have $\# \phi^{-1}(S)=n_i$.

For a polytope $\mathfrak{P}\subseteq\mathcal{H}(n_1,\dots,n_m)$ we consider the linear space of polynomials $V(\mathfrak{P})$ $=\langle\{ x_1^{\alpha_1}\cdots x_m^{\alpha_m} : (\alpha_1,\dots,\alpha_m)\in \mathfrak{P}\}\rangle$ and 
define the {\em Affine Variety} code $\mathcal{C}$ as the image of the evaluation map, $\mathcal{C}=\mathcal{C}(\mathcal{A},V(\mathfrak{P}))=\mbox{ev}_{\mathcal{A}}(V(\mathfrak{P}))$.
Then $\mathcal{C}$  is a code of length $n$ and dimension $\# \mathfrak{P}$ by Corollary \ref{corprod}. Note that in the language of Section \ref{section3}, we take $\mathcal{P}=\mathcal{A}$. 

Now we intend to give some properties on $\mathfrak{P}$ in order to obtain LRC codes whose local recovery is based on polynomial interpolation.  Define the degree of $\mathfrak{P}$ with respect to $i$ (or $i$-degree of $\mathfrak{P}$) to be the maximum $i$-degree of a polynomial in $V(\mathfrak{P})$.

\begin{theorem}
If $\mathfrak{P}$ has $i$-degree $r-1$ with $1\le r\le n_i-1$  for some $i$, then the code $\mathcal{C}(\mathcal{P},V(\mathfrak{P}))$ is a locally recoverable code of locality $r$.   The local recovery of an erased symbol $f(P)$, $P\in\mathcal{P}$, may be performed by  Lagrangian interpolation at any $r$ other points in the fibre $\phi^{-1}(\phi(P))$.
\end{theorem}
\begin{proof}
Note that
the fibre $\phi^{-1}(\phi(P))$ has $n_i-1$ points different from $P$.    Write
$
V=V(\mathfrak{P})=\bigoplus_{j=0}^{r-1} V_j \; x_i^j
$,
with $V_j\subset \mathbb{F}_{q}[x_1,\dots,x_{i-1},x_{i+1},\dots,x_m]$.  For $f\in V(\mathfrak{P})$  we can write $f=\sum g_jx_i^j$ with $g_j\in V_j$.  Observe that these  $g_j$'s are all constant in the fibre $\phi^{-1}(\phi(P))$ and thus $f(P)$ can be recovered  by interpolating a univariate polynomial in the indeterminate $x_i$.
\end{proof}
  
The same approach provides locally recoverable codes for multiple errors and  enables us to consider LRC codes with availability.
 
\begin{corollary}
If $\mathfrak{P}$ has $i$-degree $n_i-j$ for some $i,j$, with  $j\ge 2$,  then the code $\mathcal{C}(\mathcal{P},V(\mathfrak{P}))$ is a
locally recoverable code for multiple errors,  of locality $r =n_i-1$ and capability $\rho= j$.  The local recovery of $\rho$ erased symbols  may be performed by  Lagrangian interpolation at the other $n_i-j+1$ points in the fibre $\phi^{-1}(\phi(P))$.
\end{corollary}

\begin{corollary}
If $\mathfrak{P}$ has $i_j$-degree $r-1$ for some $i_1,\dots, i_t$, and $1\le r\le n_{i_j}-1$, then the code $\mathcal{C}(\mathcal{P},V(\mathfrak{P}))$ is a
locally recoverable code with $t$-availability and locality $r$ for all  $i_1,\dots, i_t$.   
\end{corollary}  
   
The most obvious choice to produce LRC Affine Variety codes  is to take $i=m$ and the polytope  $\mathfrak{P}=\mathcal{H}(n_1,n_2,\dots,n_m-1)$, which gives a code $\mathcal{C}=\mathcal{C}(\mathcal{P}, V(\mathfrak{P}))$  of locality $r=n_m-1$. 
>From Corollary  \ref{corprod},  the minimum distance of  $\mathcal{C}$ is $2$ so it is an optimal trivial code. Subsequent manipulations of that polytope also provide interesting codes.  Note that when 
$\mathfrak{P}$  is a hypercube in $\mathcal{H}(n_1,n_2,\dots,n_m-1)$ then the minimum distance of  $\mathcal{C}(\mathcal{P}, V(\mathfrak{P}))$ is also given by Corollary \ref{corprod}.

\begin{example}\label{exaff2}
Consider the field $\mathbb{F}_7$, the ring of polynomials in two variables $\mathbb{F}_7[x,y]$ and the projection $\phi=x: \mathbb{A}^2\rightarrow \mathbb{A}^1$. Let us show some LRC codes obtained by the previous construction. The following  Table \ref{texaff2} contains the parameters of nontrivial LRC codes obtained by using the following polytopes:
$\mathfrak{P}_1=\mathcal{H}(2,2)\setminus \{ (1,1)\}$,
$\mathfrak{P}_2=\mathcal{H}(3,2)\setminus \{ (2,1)\}$,
$\mathfrak{P}_3=\mathcal{H}(2,5)\setminus \{ (1,4)\}$,
$\mathfrak{P}_4=\mathfrak{P}_3\setminus \{ (1,3)\}$,
$\mathfrak{P}_5=\mathfrak{P}_4\setminus \{ (1,2),(0,4)\}$,
$\mathfrak{P}_{6}=\mathcal{H}(3,5)\setminus \{ (2,4)\}$,
$\mathfrak{P}_{7}=\mathfrak{P}_{6}\setminus \{ (2,3)\}$.
The minimum distances of these codes have been computed with Magma. The codes corresponding to the polytopes $\mathfrak{P}_3,\mathfrak{P}_4$ have the best known parameters according to the  tables \cite{CodeTables}.

\begin{table}[htbp]
\small
\begin{center}
\begin{tabular}{ccccc}
$n_1,n_2$ & polytope & parameters & locality & remarks \\
\hline 
$2,3$ & $\mathfrak{P}_1$ & $[6,3,3]$ & $2$ &  optimal  \\ 
$3,3$ & $\mathfrak{P}_2$ & $[9,5,3]$ & $2$ & optimal  \\ 
$2,6$ & $\mathfrak{P}_3$ & $[12,9,3]$ &5 & optimal  \\ 
$2,6$ & $\mathfrak{P}_4$ & $[12,8,4]$ &5 & optimal  \\ 
$2,6$ & $\mathfrak{P}_5$ & $[12,6,5]$ & 4 & defect 1  \\ 
$3,6$ & $\mathfrak{P}_{6}$ & $[18,14,3]$ & 5 & optimal  \\ 
$3,6$ & $\mathfrak{P}_{7}$ & $[18,13,4]$ & 5 & optimal  \\ \hline
\end{tabular}
\vspace*{-1mm}
\caption{\small LRC Affine Variety codes of Example \ref{exaff2}}
\label{texaff2}
\end{center}
\end{table}
\end{example}

\begin{example}\label{exaff3}
Consider again the field $\mathbb{F}_7$, the ring of polynomials in three variables $\mathbb{F}_7[x,y,z]$ and the projection $\phi=(x,y): \mathbb{A}^3\rightarrow \mathbb{A}^2$.  The following  Table \ref{texaff3} contains the parameters of  nontrivial LRC codes obtained by using the following polytopes:
$\mathfrak{Q}1=\mathcal{H}(2,2,2)\setminus \{ (1,1,1)\}$,
$\mathfrak{Q}_2=\mathfrak{Q}_1\setminus \{ (1,1,0)\}$
$\mathfrak{Q}_3=\mathcal{H}(2,2,5)\setminus \{ (1,1,4)\}$,
$\mathfrak{Q}_4=\mathfrak{Q}_3\setminus \{ (1,1,3)\}$,
$\mathfrak{Q}_5=\mathcal{H}(3,3,2)\setminus \{ (2,2,1)\}$,
$\mathfrak{Q}_6=\mathfrak{Q}_5\setminus \{ (2,2,0)\}$,
$\mathfrak{Q}_7=\mathcal{H}(3,3,5)\setminus \{ (2,2,4)\}$,
$\mathfrak{Q}_8=\mathfrak{Q}_7\setminus \{ (2,2,3)\}$.

\begin{table}[htbp]
\small
\begin{center}
\begin{tabular}{ccccc}
$n_1,n_2,n_3$ & polytope & parameters & locality & remarks \\
\hline 
$2,2,3$ & $\mathfrak{Q}_1$ & $[12,7,3]$ & 2 & optimal  \\
$2,2,3$ & $\mathfrak{Q}_2$ & $[12,6,4]$ & 2 & almost-optimal  \\ 
$2,2,6$ & $\mathfrak{Q}_3$ & $[24,19,3]$ & 5 & optimal  \\ 
$2,2,6$ & $\mathfrak{Q}_4$ & $[24,18,4]$ & 5 & optimal  \\ 
$3,3,3$ & $\mathfrak{Q}_5$ & $[27,17,3]$ & 2 & optimal  \\ 
$3,3,3$ & $\mathfrak{Q}_6$ & $[27,16,4]$ & 2 & almost-optimal  \\ 
$3,3,6$ & $\mathfrak{Q}_7$ & $[54,44,3]$ & 5 & optimal  \\ 
$3,3,6$ & $\mathfrak{Q}_8$ & $[54,13,4]$ & 5 & optimal  \\  \hline
\end{tabular}
\vspace*{-1mm}
\caption{LRC Affine Variety codes of Example \ref{exaff3}}
\label{texaff3}
\end{center}
\end{table}
\end{example}

\subsection{LRC codes from Toric codes}
\label{toric}

Affine Variety codes with $n_i=q-1$, $i=1,\dots,m$, are called {\em Toric codes}. In this case $n=(q-1)^m$ and $\mathcal{A}=\mathcal{P}=(\mathbb{F}_{q}^*)^n$.   Toric codes are rather well known and this family includes some codes with good parameters, \cite{toric1}.
  
As in the case of  other polynomial codes, the most obvious choice to produce LRC Toric codes  is to take a polytope  $\mathfrak{P}$  `close' to the hypercube $\mathcal{H}(q-1,\dots,q-1,q-2)$ or to the simplex $\Delta(l)$, $l\le q-1$. Recall that $\# \Delta(l)=\binom{m+l}{m}$ and  the minimum distance of  $\mathcal{C}(\mathcal{P},V(\Delta(l)))$ is known to be   $(q-1)^{m-1}(q-l-1)$, \cite{toric1}.

\begin{example}\label{toric2d}
Consider the field $\mathbb{F}_7$, the rings of polynomials in two and three variables  and the projection $\phi=x$.  The following  Table \ref{Ttoric2d} contains the parameters of  LRC codes obtained by using the following polytopes: 
$\mathfrak{P}_{1}=\mathcal{H}(6,5)\setminus \{ (5,4)\}$,
$\mathfrak{P}_{2}=\mathfrak{P}_{1}\setminus \{ (5,3)\}$,
$\mathfrak{P}_{3}=\mathfrak{P}_{2}\setminus \{ (5,2),(4,4)\}$, 
$\mathfrak{P}_{4}=\mathfrak{P}_{3}\backslash\{(5,1)\}$,  
$\mathfrak{R}_{1}=\Delta(4)\backslash\{(4,0),(0,4)\}$ 
in dimension two; and
$\mathfrak{Q}_{1}=\mathcal{H}(6,6,5)\setminus \{ (5,5,4)\}$,
$\mathfrak{Q}_{2}=\mathcal{Q}_{1}\setminus \{ (0,0,0)\}$
in dimension three. For a better understanding of these results we recall that the best known  $[36,25]$ and $[36,13]$ codes over $\mathbb{F}_7$ have minimum distances 7 and 17 respectively, \cite{CodeTables}.     
      
\begin{table}[htbp]
\small
\begin{center}
\begin{tabular}{cccc}
 polytope & parameters & locality & remarks \\
\hline 
 $\mathfrak{P}_{1}$ & $[36,29,3]$   & 5 & optimal  \\ 
 $\mathfrak{P}_{2}$ & $[36,28,4]$   & 5 & optimal  \\ 
 $\mathfrak{P}_{3}$ & $[36,26,5]$   & 5 & defect 1  \\ 
 $\mathfrak{P}_{4}$ & $[36,25,6]$   & 5 & defect 2  \\ 
 $\mathfrak{R}_{1}$ & $[36,13,15]$ & 4 & availability 2  \\ 
                               &                     &    & corrects 2 erasures  \\ 
 $\mathfrak{Q}_{1}$ & $[ 216, 179, 3 ]$ & 5 & optimal  \\ 
 $\mathfrak{Q}_{2}$ & $[ 216, 178, 4 ]$ & 5 & optimal  \\  \hline
\end{tabular}
\vspace*{-2mm}
\caption{LRC Toric codes of Example \ref{toric2d}}
\label{Ttoric2d}
\end{center}
\end{table}   
\end{example}

\subsection{LRC codes from Reed-Muller codes}\label{RM}

Let  $\mathcal{R}[x_1,\dots,x_m]$ be the set of reduced polynomials in $\mathbb{F}_{q}[x_1,\dots,x_m]$ and $\mathcal{R}[x_1,\dots,x_m]_l$ be the set of reduced polynomials of total degree at most $l$.  
Note that these sets  are the spaces of functions $V(\mathfrak{P})$ associated to the polytopes $\mathfrak{P}=\mathcal{H}(q,\dots,q)$ and
$\mathfrak{P}=\Delta(l)$, respectively.  It is well known that the evaluation at $\mathcal{A}=\mathbb{A}^m(\mathbb{F}_{q})$ map  $\mbox{ev}_{\mathcal{A}}:\mathcal{R}\rightarrow\mathbb{F}_{q}^n$ is an isomorphism (\cite{RM} or  Corollary \ref{corprod}).  The Reed-Muller code
RM$(l,m)$ is defined as RM$(l,m)=\mbox{ev}_{\mathcal{A}}(\mathcal{R}[x_1,\dots,x_m]_l)$.  When $l\le q-2$ this is already an LRC code. To see that we take  $\mathcal{S}=\mathbb{A}^{m-1}(\mathbb{F}_{q})$,  the map $\phi=(x_1,\dots,x_{m-1}) : \mathbb{A}^m(\mathbb{F}_{q})\rightarrow \mathbb{A}^{m-1}(\mathbb{F}_{q})$ and $\mathcal{P}=\mathcal{A}$. It is clear that for any $S\in\mathbb{A}^{m-1}(\mathbb{F}_{q})$ the fibre $\phi^{-1}(S)$ consists of $q$ points, hence RM$(l,m)$ has locality $r=q-1$ and length  $n=\#\mathbb{A}^m(\mathbb{F}_{q})=q^m$.   When $l\le q-1$ then  RM$(l,m)$ has minimum distance $d=(q-l)q^{m-1}$ and dimension is $\#\Delta(l)=\binom{m+l}{m}$. 

In general, given a polytope $\mathfrak{P}\subseteq \mathcal{H}(q,\dots,q,q-1)$ we can consider the linear space $V(\mathfrak{P})$, which can be written as $V(\mathfrak{P})=\bigoplus_{i=0}^{q-2} V_i \, x_m^i$ with  $V_i\subseteq \mathcal{R}[x_1,\dots,x_{m-1}]$, and the code  $\mathcal{C}(\mathcal{P},V(\mathfrak{P}))=\mbox{ev}_{\mathcal{P}}(V(\mathfrak{P}))$.
This is an LRC code on length $n=q^m$, locality $r\le q-1$ and dimension $k=\# \mathfrak{P}$, since $\mbox{ev}_{\mathcal{P}}$ in injective.  The next  propositions state two interesting properties of such codes $\mathcal{C}(\mathcal{P},V(\mathfrak{P}))$.

\begin{proposition}
Let $\mathfrak{P}\subseteq\mathcal{H}((q,\dots,q,q-1)$ be a polytope  and let $\mathcal{C}=\mathcal{C}(\mathcal{P},V(\mathfrak{P}))$.  For any codeword $\mathbf{c}\in \mathcal{C}$ the sum of the coordinates of $\mathbf{c}$ corresponding to each fibre of $\phi$ is zero. Thus $\mathcal{C}$ allows a local recovery of single erasures by one addition.
\end{proposition}
\begin{proof}
Let $f\in V(\mathfrak{P})$ and $P\in\mathcal{P}$. Write $\phi(P)=S$ and $\phi^{-1}(S)=\{ P_1,\dots,P_{q}\}$ the fibre of $S$, with $P=P_j$ for some $j$.  The points $ P_1,\dots,P_{q}$, differ in their $m$-th coordinate, which runs through all the elements of $\mathbb{F}_{q}$. Then, since  $f$ acts on the fibre  $\phi^{-1}(S)$ as a polynomial $f_S (x_m)$, we have
$$
\sum_{i=1}^q f(P_i)=\sum_{\alpha\in\mathbb{F}_{q}} f_S(\alpha)=0
$$
where the last equality follows from Lemma \ref{lemma} as the elements $\alpha\in\mathbb{F}_{q}$ are precisely the roots of the linearized polynomial $x^q-x$ (or see \cite{RM}).
\end{proof} 

The second interesting property of these codes refers to the availability problem. Fix a variable $x_j$. Let $\delta$ be the maximum degree in $x_j$ of all elements in $V(\mathfrak{P})$. Then this space can be written as 
$V(\mathfrak{P})=\bigoplus_{i=0}^{\delta} V'_i \, x_j^i$     
with $V'_i\subseteq\mathcal{R}[x_1,\dots,x_{j-1},x_{j+1},\dots,x_{m}]$.
If $\delta\le q-2$, an erasure at the position $f(P)$ of the word $\mbox{ev}_{\mathcal{P}}(f)$ can be recovered from the fibre  $\psi^{-1}(\psi(P))$, where $\psi$ is the projection map
$\psi=(x_1,\dots,x_{j-1},x_{j+1},\dots,x_{m}) : \mathbb{A}^m(\mathbb{F}_{q})\rightarrow \mathbb{A}^{m-1}(\mathbb{F}_{q})$.  Note that the fibres 
$\phi^{-1}(\phi(P))$ and $\psi^{-1}(\psi(P))$ only meet at $P$.  In particular, when $\mathfrak{P}\subseteq \mathcal{H}(q-1,\dots,q-1)$,   then each coordinate of a codeword in $\mathcal{C}(\mathcal{P},V(\mathfrak{P}))$ has $m$ disjoint recovering sets. We have the following result.
  
\begin{proposition}
Let $\mathfrak{P}\subseteq\mathcal{H}((q,\dots,q)$ be a polytope  and let $\mathcal{C}=\mathcal{C}(\mathcal{P},V(\mathfrak{P}))$. 
If $\mathfrak{P}$ has $i$-degree $r-1$  for some $i_1,\dots, i_t$ and  $1\le r\le q-1$, then the code $\mathcal{C}(\mathcal{P},V(\mathfrak{P}))$ is an LRC code with $t$-availability and locality $r$ for all  $i_1,\dots, i_t$.   
\end{proposition}

\begin{example}
Let $q=7$, $\mathcal{P}=\mathbb{A}^2({\mathbb{F}_7})$ and $\mathfrak{P}\subseteq \Delta(6)\subset \mathbb{N}_0^2$. \newline
(a) If $\mathfrak{P}=\Delta(6)$ then $\mathcal{C}(\mathcal{P},\mathfrak{P})$ is the Reed-Muller code $\mbox{RM}(6,2)$ of parameters $[49,28,7]$. Note that this code is not included in our construction as LRC code. In fact $\mbox {RM} (6,2) $ has locality $r\ge d(\mbox{RM}(6,2)^{\perp})-1=d(\mbox{RM}(5,2))-1=13$. Let
$\mathfrak{P}_1=\Delta(6)\setminus \{ (0,6),(6,0)\}$. Then $\mathcal{C}_1=\mathcal{C}(\mathcal{P},\mathfrak{P}_1)$ is a $[49,26,12]$ code. The best known minimum distance for a $[49,26]_7$ linear code is  $d=14$, \cite{CodeTables}. Let
$\mathfrak{P}_2=\mathfrak{P}_1\setminus \{ (1,1)\}$. Then $\mathcal{C}_2=\mathcal{C}(\mathcal{P},\mathfrak{P}_2)$ is a $[49,25,14]$  code. Both $\mathcal{C}_1$ and $\mathcal{C}_2$ are LRC codes of locality $r=6$ for which any coordinate has two disjoint recovering sets, and the recovery can be performed through checksum. \newline
(b) If $\mathfrak{P}=\Delta(5)$ then $\mathcal{C}=\mathcal{C}(\mathcal{P},\mathfrak{P})$ is the Reed-Muller code $\mbox{RM}(5,2)$ of parameters $[49,21,14]$. This is an LRC code of locality $r=6$. Let 
$\mathfrak{P}_1=\Delta(6)\setminus \{ (0,5),(5,0)\}$. Then $\mathcal{C}_1=\mathcal{C}(\mathcal{P},\mathfrak{P}_1)$ is a $[49,19,18]$ code. The best known minimum distance for a $[49,19]_7$ linear code is  $d=20$, \cite{CodeTables}. Let
$\mathfrak{P}_2=\mathfrak{P}_1\setminus \{ (1,1)\}$. Then $\mathcal{C}_2=\mathcal{C}(\mathcal{P},\mathfrak{P}_2)$ is a $[49,18,20]$  code.
The best known minimum distance for a $[49,18]_7$ linear code is  $d=21$, \cite{CodeTables}. 
Both $\mathcal{C}_1$ and $\mathcal{C}_2$ are LRC codes of locality $r=5$. As in (a), any coordinate has two disjoint recovering sets, and the recovery can be performed through checksum.   
\end{example}
   
\section*{Acknowledgments}

The first author was supported by Spanish Ministerio de Econom\'{\i}a y Competitividad  under grant MTM2015-65764-C3-1-P MINECO/FEDER. The second author was supported by CNPq-Brazil under grants 159852/2014-5 and 201584/2015-8.


\begin{thebibliography}{00}

\bibitem{ballico}  
E. Ballico, C. Marcolla,
Higher Hamming weights for locally recoverable codes on algebraic curves,
Finite Fields Appl.  40  (2016), 61--72.

\bibitem{barg2}  
A. Barg, I. Tamo, S. Vladut, 
Locally recoverable codes on algebraic curves,
in Proceedings of ISIT-2015, 
Hong Kong,  2015, 1252--1256.

\bibitem{barg3}  
A. Barg, I. Tamo, S. Vladut,
Locally recoverable codes on algebraic curves, 
ArXiv:1603.08876.

\bibitem{Vieta}
D. Cox, J. Little, D. O'Shea, 
Ideals, Varieties, and Algorithms,
Springer, New York, 1992. 

\bibitem{AfineVariey}
O. Geil,
Evaluation codes from an Affine Variety code perspective,
in  E. Martinez (Ed.), Advances in Algebraic Geometry codes,  World Scientific, Hackensack, 2008,  153--180.

\bibitem{gopalan} 
P. Gopalan, C. Huang, H. Simitci, S. Yekhanin,
On the locality of codeword symbols,
IEEE Trans. Inform. Theory  58(11)  (2012),  6925--6934. 

\bibitem{CodeTables}
M. Grassl, 
Bounds on the minimum distance of linear codes.
Online available at http://www.codetables.de. Accessed on 2017-05-15.

\bibitem{Klein}
J. Hansen,
Codes on the Klein quartic, ideals and decoding, 
IEEE Trans. Inform. Theory 33(6) (1987),  923--925. 

\bibitem{matthews}
K. Haymaker, B. Malmskog, G.L. Mathews, Locally recoverable codes with availability $t\geq 2$ from fiber products of curves. ArXiv: 1612.03841, 2016.

\bibitem{RM}
T. H\o holdt, J.H. van Lint,  R. Pellikaan,
Algebraic geometry codes,
in V.S. Pless, W.C. Huffman, R.A. Brualdi (Eds.), Handbook of Coding Theory, 
Elsevier, Amsterdam 1998, 871--961.

\bibitem{toric1} 
J. Little, R. Schwarz,
On toric codes and multivariate Vandermonde matrices,
Appl. Algebra Engrg. Comm. Comput.  18(4)  (2007),  349--367.

\bibitem{magma}
Magma Computational Algebra System. Online available at http://magma. maths.usyd.edu.au/magma/.

\bibitem{maharaj} 
H. Maharaj, 
Explicit constructions of algebraic-geometric codes, 
IEEE Trans. Inform. Theory 51(2) (2005),  714--722.

\bibitem{Rendiconti}
C. Munuera,
An algorithm to compute the number of points on elliptic curves of j-invariant 0 or 1728 over a finite field,
Rend. Circ. Mat. Palermo (II)  42(1)  (1993), 106--116.

\bibitem{availability}
A.S. Rawat, D.S. Papailiopoulos, A.G. Dimakis, S. Vishwanath, 
Locality and availability in distributed storage, 
in Proceedings of ISIT-2014, Honolulu, 2014,  681--685.

\bibitem{largo1}      
N. Silberstein, A. Singh Rawat,  S. Vishwanath, 
Error-correcting regenerating and locally repairable codes via rank-metric codes, 
IEEE Trans. Inform. Theory 61(11) (2015), 5765--5778.

\bibitem{prod}
I. Soprunov, J. Soprunova,
Bringing toric codes to the next dimension,
SIAM J. Discrete Math.  24(2)  (2010),  655--665.  

\bibitem{barg1} 
I. Tamo,  A. Barg, 
A family of optimal locally recoverable codes,
IEEE Trans. Inform. Theory 60(8)  (2014),  4661--4676.

\end{thebibliography}
\end{document}